\documentclass[a4paper]{article}

\usepackage{amsthm,amsmath,amsfonts}

\newtheorem{thm}{Theorem}
\newtheorem{cor}[thm]{Corollary}
\newtheorem{lem}[thm]{Lemma}

\newcommand{\RR}{{\mathbb R}}
\newcommand{\CC}{{\mathbb C}}
\newcommand{\NN}{{\mathbb N}}

\newcommand{\Sch}{{\mathcal S}}
\newcommand{\FF}{{\mathcal F}}
\newcommand{\HH}{{\mathcal H}}
\newcommand{\VV}{{\mathcal V}}

\newcommand{\varkappa}{\kappa}

\DeclareMathOperator{\spec}{spec}

\DeclareMathOperator{\supp}{supp}

\renewcommand{\det}{{\rm det}}

\newcommand{\mR}{\mathbb{R}}
\newcommand{\mC}{\mathbb{C}}
\newcommand{\cH}{\mathcal{H}}
\newcommand{\cF}{\mathcal{F}}

\begin{document}

\title{\bf On the discrete spectrum of spin-orbit Hamiltonians
with singular interactions}

\author{\sc Jochen Br\"uning\dag \and \sc Vladimir Geyler\dag\ddag\thanks{Deceased
on April 2, 2007} \and \sc Konstantin Pankrashkin\dag\S\\[\bigskipamount]
\dag Institut f\"ur Mathematik, Humboldt-Universit\"at zu Berlin\\
Rudower Chaussee 25, 12489 Berlin, Germany\\
\ddag Mathematical Faculty, Mordovian State University\\
430000 Saransk, Russia\\
\S D\'epartement de Math\'ematiques,
Universit\'e Paris Nord\\ 99 av. J.-B. Cl\'ement, 93430 Villetaneuse, France\\
Corresponding author,
e-mail: const@math.hu-berlin.de}

\date{}

\maketitle

\begin{abstract}
We give a variational proof of the existence of infinitely many
bound states below the continuous spectrum for spin-orbit Hamiltonians
(including the Rashba and Dresselhaus cases) perturbed by measure potentials
thus extending the results of J.~Br\"uning, V.~Geyler, K.~Pankrashkin:
J. Phys. A: Math. Gen. {\bf 40} (2007) F113--F117.
\end{abstract}

\section{Introduction}

It is well known that perturbations of the Laplacian
by sufficiently localized potentials produce only a finite number
of negative eigenvalues, and only long-range potentials can produce
an infinite discrete spectrum, see Section XII.3 in \cite{RS4}. 
This does not hold any more if one considers perturbations of magnetic Schr\"odinger
operators, where compactly suppported perturbations demonstrate a non-classical
behavior~\cite{RW,RT}.

Recently such questions have become of interest in the study
of operators related to the spintronics. Namely, in \cite{CM}
it was emphasized that perturbations of the Rashba and Dresselhaus
Hamiltonians
\begin{gather*}
H_R=\begin{pmatrix}
p^2 & \alpha(p_y+ip_x)\\
\alpha(p_y-ip_x) & p^2
\end{pmatrix},\\
H_D=
\begin{pmatrix}
p^2 & -\alpha(p_x+i p_y)\\
-\alpha(p_x-ip_y) & p^2
\end{pmatrix},
\end{gather*}
($\alpha$ is a constant expressing the strength of the
spin-orbit coupling \cite{BR,RS,Win})
by localized spherically symmetric negative potentials produce infinitely
many eigenvalues below the continuous spectrum; in the justification
some approxations have been used.
In the paper \cite{BGP1} we gave a rigoruous proof of this effect
for rather general negative potentials without any symmetry conditions.
In the present note we are going to extend these results and to obtain similar
estimates for operators of the form $H=H_0+\nu$, where
$H_0$ is an unperturbed spin-orbit Hamiltonian and
$\nu$ is a measure (whose support can have the zero Lebesgue measure).
In particular, for the Rashba and Dresselhaus Hamiltonians
we show that negative perturbations
supported by curves  always produce infinitely many bound states below
the threshold.

\section{Definition of Hamiltonians}

Denote by $\HH$ the Hilbert space
$L^2(\RR^2)\otimes\CC^2$ of two-dimensional spinors; by $\FF$ we
denote the Fourier transform $\FF:\,L^2(\RR^2)\rightarrow
L^2(\RR^2)$; then $\FF_2:=\FF\otimes 1$ is the Fourier
transform in $\HH$. 
Let $H_0$ be the self-adjoint operator in
$\cH$ whose Fourier transform $\widehat H_0:=\cF_2H_0\cF_2^{-1}$ is the
multiplication by the matrix
\begin{equation}
                \label{S1}
\widehat H_0(p)=\begin{pmatrix} p^2& A(p)\\
                  \overline{A(p)}& p^2\\
                  \end{pmatrix}\,,\quad p\in\mR^2\,,
\end{equation}
where $A$ is a continuous complex function on $\mR^2$.
We assume
\begin{equation}
\label{A1} \limsup\limits_{p\to \infty}\dfrac{|A(p)|}{p^2}<1.
\end{equation}

Clearly, $H_0$ has no discrete spectrum; its spectrum is the union
of images of two functions $\lambda_{\pm}$ (dispersion laws):
$\lambda_\pm(p)=p^2\pm|A(p)|$, hence $\spec
H_0=[\varkappa,+\infty)$, where
$\varkappa:=\inf\{p^2-|A(p)|\,:\,p\in\mR^2\}>-\infty$.
Moreover, there is a unitary matrix $M(p)$ depending
continuously on $p\in\mR^2$ such that
\begin{equation}
                \label{S3}
M(p)\widehat H_0(p)M^*(p)=
                 \begin{pmatrix} \lambda_+(p)&0\\
                  0& \lambda_-(p)\\
                  \end{pmatrix},\quad\,p\in\mR^2\,.
\end{equation}
Denote $S:=\{p\in\mR^2:\,\lambda_-(p)=\varkappa\}$; this is a
non-empty compact set. We will assume that
\begin{equation}
\label{A2}
\text{the function $|A(p)|$ is of class $C^2$ in a
neighborhood of $S$.}
\end{equation}
For the Rashba and Dresselhaus Hamiltonians one has
$\varkappa=-\alpha^2/4$  and $S$ is the circle
$\{p: 2|p|=|\alpha|\}$, and the condition (\ref{A2}) is obviously satisfied.

The two condtions \eqref{A1} and \eqref{A2} guarantee
that for any $p_0\in S$ there is a constant $c(p_0)>0$ such that
\begin{equation}
                \label{S4}
0\le\lambda_-(p)-\varkappa\le c(p_0)|p-p_0|^2\quad
\text{for all }p\in\mR^2.
\end{equation}

We fix a positive Radon measure $m$ on $\RR^2$ and a bounded
Borel measurable function $h:\RR^2\to\RR$ such that there exist
constants $a\in (0,1)$ and $b>0$ with
\begin{equation}
         \label{eq-inf}
\int_{\RR^2} (1+|h(x)|^2)|f(x)|^2 m(dx)\le
a\int_{\RR^2} |\nabla f(x)|^2dx+b\int_{\RR^2} |f(x)|^2dx
\end{equation}
for all $f$ from the Schwartz space $\Sch(\RR^2)$.
Denote $\nu:=hm$.
The assumption \eqref{eq-inf} is satisfied for any (bounded) $h$
if $m$ belongs to the Kato class measures, i.e.
\[
\lim_{\varepsilon\to 0+} \sup_{x\in\RR^2}
\int_{|x-y|<\varepsilon}
\log \dfrac{1}{|x-y|}\big| m(dy)=0.
\]
This holds, for example, for $\delta$-type measures concentrated on $C^1$ curves
under some regularity conditions (this conditions are satisfied for compact curves and straight lines),
see Section 4 in \cite{BEKS} for details.

Our aim now is to give a rigorous definition of the operator
given by the formal expression
\begin{equation}
    \label{eq-hhm}
H=H_0+\nu.
\end{equation}

As $\Sch(\RR^2)$ is dense in the Sobolev space $H^1(\RR^2)$
there exists a unique linear bounded transformation $J$
defined by
\[
J:H^1(\RR^2)\to L^2(\RR^2,m), \quad
J f=f \text{ for all } f\in\Sch(\RR^2).
\]
We denote $J_2:=J\otimes 1$;
this is an operator acting from $H^1(\RR^2)\otimes\CC^2$ to $L^2(\RR^2,m)\otimes\CC^2$.
For a continuous function $f$ we denote the corresponding
equivalence classes in $L^2(\RR)$ and $L^2(\RR, m)$
by the same letter $f$.

Now note that the operator $H_0$ can be presented as
\[
H_0=-\Delta_2  + L, \quad
-\Delta_2:=-\Delta\otimes 1, \quad
\widehat L:=\cF_2 L\cF_2^{-1}=\begin{pmatrix} 0& A(p)\\
                  \overline{A(p)}& 0\\
                  \end{pmatrix}.
\]
(Here $\Delta$ is the scalar two-dimensional Laplacian.)
\begin{lem}\label{lem1}
The operator $L$ is relatively bounded with respect to $\Delta_2$, and
$\|L(-\Delta_2+\lambda)^{-1}\|<1$ for $\lambda\to+\infty$.
\end{lem}

\begin{proof} The relative boundedness is obvious, so we only prove the norm estimate.
Passing to the Fourier transform, we need to show
\begin{equation}
             \label{eq-an}
\sup_{p\in\RR^2} \Big|
\dfrac{A(p)}{p^2+\lambda}
\Big|<1, \quad \lambda\to+\infty.
\end{equation}

By \eqref{A1}, there exist $a<1$ and $R>0$ such that
$|A(p)|/ p^2\le a$ for all $p$ with $|p|>R$.
Then obviuosly one has 
\begin{equation}
   \label{ea-estz}
\Big|
\dfrac{A(p)}{p^2+\lambda}
\Big|<a,\quad |p|>R, \quad \lambda>0.
\end{equation}
Due to the continuity of $A$ there exists $C>0$ with $|A(p)|\le C$
for $|p|\le R$. Then obviously there exists $\lambda_0>0$ such that
\begin{equation}
   \label{ea-estz2}
\Big|
\dfrac{A(p)}{p^2+\lambda}
\Big|\le C \lambda^{-1}<a,\quad |p|\le R, \quad \lambda>\lambda_0.
\end{equation}
Combining \eqref{ea-estz} with  \eqref{ea-estz2} we arrive at \eqref{eq-an}.
\end{proof}

Eq. \eqref{eq-inf} and Lemma~\ref{lem1} imply that, by the KLMN theorem,
the quadratic form
\begin{multline*}
q(f,g)=q_0(f,g)+ \nu(f,g),\\
\nu(f,g):=\langle h J_2 f, J_2 g \rangle_{L^2(\RR^2)\otimes\CC^2}\equiv
\int_{\RR^2} \langle J_2 f (x), J_2 g (x)\rangle_{\CC^2} \nu(dx),\\
\nu(dx)= h(x) m(dx),
\end{multline*}
where $q_0$ is the quadratic form associated with $H_0$,
is semibounded below and closed on $H^1(\RR^2)\otimes\CC^2$
and hence defines a certain self-adjoint operator $H$
semibounded below. If the measure $\nu$ is absolutely continuous
with respect to the Lebesgue measure, i.e. $\nu(dx)= V(x)dx$
with a certain locally integrable function $V$, then the above procedure
gives the usual form sum $H=H_0+V$, so one preserves the same notation for the general case,
$H_0+\nu:=H$.

Repeating the procedure from Section~2 in \cite{BEKS}
one can express the resolvent of $H$ through the resolvent $H_0$.
Combining this with the explicit expressions for the Green function
for $H_0$ \cite{BGP2} one can obtain rather detailed formulas
for the Green function of $H$, but we will not need this below.
The following assertion about the spectral properties of $H$ is important for us.

\begin{thm}
If the measure $m$ is finite, then the essential spectra of $H_0$
and $H$ coincide.
\end{thm}

\begin{proof}
The paper \cite{br} deals with a rather detailed spectral
analysis of operators defined by the sums of quadratic forms.
According to Theorem 7 in \cite{br} we only need to prove that
the operator $J_2 (H_0-z)^{-1}:L^2(\RR^2)\otimes\CC^2\to L^2(\RR^2,m)\otimes\CC^2$
is compact. Note that that for sufficiently large $\lambda$
the operator $1 +L(-\Delta_2+\lambda)^{-1}$ has a bounded inverse
defined everywhere due to lemma~\ref{lem1}.
Hence $(H_0+\lambda)^{-1}=(-\Delta_2+\lambda)^{-1}\big(1 +L(-\Delta_2+\lambda)^{-1}\big)^{-1}$,
and the compactness of $J_2 (H_0+\lambda)^{-1}$ would follow from
the compactness of $J (-\Delta+\lambda)^{-1}: L^2(\RR^2)\to L^2(\RR^2,m)$
as  $J_2 (-\Delta_2+\lambda)^{-1}\equiv
\big(J (-\Delta+\lambda)^{-1}\big)\otimes 1$.
At the same time, $J (-\Delta+\lambda)^{-1}=B^*$,
\[
B:L^2(\RR^2,m)\to L^2(\RR^2),\quad
B f (x)=  \int_{\RR^2} G_0(x,y;\lambda) f(y) m(dy) \quad \text{a.e.},
\]
where $G_0$ is the Green function of the two-dimensional Laplacian.
In Lemma~2.3 in \cite{BEKS} it was shown that $B$ is compact.
Therefore, $J (-\Delta+\lambda)^{-1}$ is also compact, and the theorem is proved.
\end{proof}

{}From now on we assume that the measure $m$ is finite.

\section{General perturbations}

Below for a distribution $f$ we denote the Fourier transform of $f$ by
$\widehat f$. We call an Hermitian $n\times n$ matrix $C$
\emph{positive definite} if for any non-zero $\xi\in\mC^n$ there holds $\langle \xi,C\xi\rangle>0$,
and \emph{positive semi-definite} if the above equality is non-strict.
By analogy one introduces \emph{negative definite} and \emph{negative semi-definite}
matrices.

The following result differs only in minor details from the main result
in \cite{BGP2}

\begin{thm}\label{th1}
Let $N\in\NN$; assume that the Fourier transform
$\widehat \nu$ satisfies the following
condition: there are $N$ points $p_1,\ldots p_N\in S$ such that
the matrix $\big(\widehat \nu(p_j-p_k)\big)_{j,k=1}^N$ is
negative definite. Then $H$ has at least $N$ eigenvalues, counting
multiplicity, below $\varkappa$.
\end{thm}

\begin{proof} According to the max-min
principle, it is sufficient to show that  we
can find $N$ vectors $\Psi_m\in\HH$, $m=1,\ldots,N$, such that the
matrix with the entries $(q-\varkappa)(\Psi_j,\Psi_k)$,
$j,k=1,\dots, N$, $(q-\varkappa)(\Phi,\Psi):=q(\Phi,\Psi)-\varkappa\langle\Phi,\Psi\rangle$
is negative definite.

Set $\displaystyle
f_a(x):=\exp\left(-\dfrac{1}{2}|x|^a\right)$, $x\in\RR^2$,
with $a>0$. Clearly, $f_a\in H^1(\RR^2)$.
It is observed in~\cite{YdL} that
\begin{equation}
              \label{S5}
\int_{\RR^2}\big|\nabla\,f_a(x)\big|^2\,dx=\frac{\pi}{2}a.
\end{equation}
Furthemore, by the Lebesgue dominated convergence theorem,
\[
\lim\limits_{a\to0+}\,\displaystyle\int_{\RR^2} |f_a(x)|^2\,\nu(dx)=e^{-1}\,\displaystyle\int_{\RR^2}
\nu(dx).
\]
Let $\widehat f_a$ be the Fourier transform of $f_a$.
Take spinors $\Psi_j$ such that their Fourier transforms
$\widehat\Psi_j$ are of the form
$\widehat\Psi_j(p)=M(p)\psi_j(p)$, where
\begin{equation}
              \label{S7}
\psi_j(p)=\left(\begin{matrix}0\\
\widehat
f_a(p-p_j)\\\end{matrix}\right)
\end{equation}
and $M(p)$ is taken from \eqref{S3}.
We show that if $a$ is sufficiently small, then the
matrix $(q-\varkappa)(\Psi_j,\Psi_k)$ is
negative definite.
For this purpose it is sufficient to show that
\begin{gather}
\label{f1}
 \lim_{a\to0} (q_0-\varkappa)(\Psi_j,\Psi_k)=0,\\
\label{f2} \lim_{a\to0}\nu (\Psi_j,\Psi_k)=2\pi
e^{-1}\widehat \nu(p_j-p_k)
\end{gather}
for all $j$ and $k$.

By definition of $\Psi_j$ one has
\begin{multline*}
\big|(q_0-\varkappa)(\Psi_j,\Psi_k)\big|=\Big|\int_{\RR^2}
(\lambda_-(p)-\varkappa)\overline{\widehat f_a(p-p_j)}
f_a(p-p_k)\,dp\Big|\\
\le\sqrt{\int_{\RR^2} \big(\lambda_-(p)-\varkappa\big)
\,\big|\widehat f_a(p-p_j)\big|^2\,dp}\\
\times
\sqrt{\int_{\RR^2}
\big(\lambda_-(p)-\varkappa\big)\,\big|\widehat
f_a(p-p_k)\big|^2\,dp}.
\end{multline*}
On the other hand, by \eqref{S4} and \eqref{S5} one has
\begin{multline*}
0\le \int_{\RR^2} \big(\lambda_-(p)-\varkappa\big)\,\big|\widehat f_a(p-p_j)\big|^2\,dp\\
\le c(p_j)\int_{\RR^2} (p-p_j)^2\,\big|\widehat f_a(p-p_j)\big|^2\,dp
= c(p_j) \int_{\RR^2} p^2\big|\widehat f_a(p)\big|^2\,dp\\
= c(p_j) \int_{\RR^2}\big|\nabla\,f_a(x)\big|^2\,dx
=\frac{\pi}{2}c(p_j)a,
\end{multline*}
which proves \eqref{f1}.
As for \eqref{f2}, one has
\begin{multline*}
\nu(\Psi_j,\Psi_k)= \int_{\RR^2}\int_{\RR^2}
\widehat \nu(p-q)
\big\langle\widehat
\Psi_j(p),\widehat\Psi_k(q)\big\rangle_{\CC^2}
dp\,dq\\
=\int_{\RR^2}\int_{\RR^2}
\widehat \nu(p-q)
\big\langle\widehat \psi_j(p),\widehat\psi_k(q)\big\rangle_{\CC^2}
dp\,dq.
\end{multline*}
On the other hand,
\begin{multline*}
\int_{\RR^2}\int_{\RR^2}
\widehat \nu(p-q)
\big\langle\widehat \psi_j(p),\widehat\psi_k(q)\big\rangle_{\CC^2}
dp\,dq\\
=\int_{\RR^2}\int_{\RR^2} \widehat \nu(p-q)\overline{\widehat f_a(p-p_j)}
\widehat f_a(q-p_k)\,dp\,dq\\
= \int_{\RR^2} e^{i\langle p_j-p_k,x\rangle}\big|f_a(x)\big|^2\,\nu(dx)
\stackrel{a\to 0}{\longrightarrow} 2\pi
e^{-1}\widehat \nu(p_j-p_k),
\end{multline*}
and the theorem is proved.
\end{proof}

The above theorem allows to generalize some classical results.
For example, the condition 
\[
\displaystyle \int_{\RR^2}
\nu(dx)\equiv 2\pi \widehat \nu(0) <0.
\]
guarantees that $H$ has at least one eigenvalue below
$\varkappa$.

To formulate further corollaries, by $\#S$ we denote
the number of points in $S$, if $S$ is finite, and $\infty$, otherwise.

\begin{cor}\label{boch} Let $\nu\le0$ and $\supp \nu$ have a positive
Lesbegue measure. Then $H$ has at least $\# S$ eigenvalues below $\kappa$.
\end{cor}

\begin{proof}
We show that the matrix $\big(\widehat \nu(p_j-p_k)\big)_{j,k=1}^N$ is
negative definite for any choice of pairwise different points
$p_1,\ldots,p_N\in\RR^2$. By the Bochner theorem,
$\displaystyle-\sum_{j,k}\widehat \nu(p_j-p_k)\bar\xi_j\xi_k\ge0$
for any $(\xi_j)\in\CC^N$ and it remains to note that
$\displaystyle\sum_{j,k}\widehat \nu(p_j-p_k)\bar\xi_m\xi_n\ne0$
for $(\xi_j)\ne0$. In fact, if $\displaystyle\sum_{j,k}\widehat
\nu(p_j-p_k)\bar\xi_j\xi_k=0$, then
\[
\int_{\RR^2}\Big|\sum_{j}\xi_j e^{i \langle p_j, x\rangle}\Big|^2\, \nu(dx)=0\,;
\]
therefore, $\displaystyle\sum_{j}\xi_j e^{i\langle p_j, x\rangle}=0$ on the
support of $\nu$. As the exponents $e^{i\langle p_j, x\rangle}$ are real
analytic, and $\displaystyle\sum_{j}\xi_j e^{i \langle p_j, x\rangle}=0$ on a
set of positive Lebesgue measure, the equality
$\displaystyle\sum_{j}\xi_j e^{i\langle p_j, x\rangle}=0$ is valid everywhere
on $\RR^2$. On the other hand, $e^{i\langle p_j, x\rangle}$ are linearly
independent, and we obtain $\xi_j=0$ for all $j$.
\end{proof}

Corollary \ref{boch} shows the existence of infintely many
eigenvalues below $\varkappa$ for perturbations of the Rashba and Dresselhaus Hamiltonians
by negative measures with support of positive Lesbegue measure,
i.e. given by negative regular potentials,
the sum of of a negative regular potential and a negative
$\delta$-measure supported by a curve etc.
Of interest is the question if corollary~\ref{boch}
still holds for the case when the support of $m$ has zero Lesbegue measure,
as the above arguments do not work in that case.

In \cite{CFS} we have found an interesting condition permitting to handle
a class of singular perturbations without assumptions on the Lebesgue measure of the support.

\begin{cor}\label{boch2} Let $\nu\le0$ and the intersection
of $\supp \nu$ with a certain circle be an infinite set.
Then $H$ has at least $\# S$ eigenvalues below $\kappa$.
In particular, this holds if $m$ is spherically symmetric and $h< 0$.
\end{cor}

\begin{proof}
In \cite{CFS} it is shown that under above assumptions
the matrix $\big(\widehat \nu(p_j-p_k)\big)$ is negative definite for any
choice of pairwise different $p_j$.
\end{proof}

\section{Perturbations supported by curves}

The above results allows for analysis of perturbations supported by curves only
in rather particular symmetric cases.
We consider in greater detail  perturbations supported by curves in this section.

\begin{lem}\label{zero}
Assume that

\begin{enumerate}
\item[(a)] $S$ contains a $C^1$ arc which is not an interval;

\item[(b)] $m$ is compactly supported;

\item[(c)] $h(x)<0$ for a.e. $x$ with respect to the measure $m$;

\item[(d)] The Fourier transform $\hat\nu\equiv\widehat{hm}$ vanishes at infinity
at least along a certain straight line, i.e. $\hat\nu(r\cos\alpha,r\sin\alpha)\to 0$ for $r\to\pm\infty$
and some $\alpha\in [0,2\pi)$.

\end{enumerate}
Then the assumptions of Theorem \ref{th1} are satisfied for any $N$.
\end{lem}

We prove first a simple lemma.

\begin{lem}\label{lem2} Let $C$ be a $C^1$ arc in the
plane $\RR^2$ different from an interval.
Then the et $X=\{p-q:\,p,q\in
C\}\subset \RR^2$ has an interior point.
\end{lem}

\begin{proof}
Let $t\mapsto \big(x(t),y(t)\big)$, $t\in[0,1]$, is a regular
parametrization of $C$. Consider the mapping $f:(s,t)\mapsto
\big(x(t)-x(s),y(t)-y(s)\big)$, $s,t\in[0,1]$.
Clear, $X=f\big([0,1]\times[0,1]\big)$.
Since $C$ is not an interval, the Jacobian
\[
\det\begin{pmatrix} -x'(s)&-y'(s)\\
 x'(t)&y'(t)
\end{pmatrix}
\]
does not vanish at a certain point
$(s_0,t_0)\in(0,1)\times(0,1)$. By the implicit function theorem,
$f$ is invertible near $(s_0,t_0)$, therefore the image of some
neighborhood of $(s_0,t_0)$ is an open set.
\end{proof}

\begin{proof}[\bf Proof of Lemma \ref{zero}]
The assumption (b) implies the analyticity of $\hat\nu$.
We show that for any $N$ there are points
$p_j\in S$, $j=1,\ldots,N$, such that the matrix
$\VV(p_1,\ldots,p_n)$ with entries
$\VV_{jk}=\hat\nu(p_j-p_k)$ is negative definite.
By the Bochner theorem, $\VV$ is negative semi-definite
for any choice of $p_j$, and it remains to
show that $\det \VV(p_1,\ldots,p_N)\ne0$ for a choice of
points $p_j$ from $S$.

Now we proceed by induction on $N$. For $N=1$ one has,
by (c), $\det \VV=\Hat\nu(0)<0$.

Let now $N\ge 2$. Assume that there exist $p'_j\in S$, $j=1,\ldots,N-1$ such that  $\det \VV(p'_1,\ldots,p'_{N-1})\ne0$.

Note that $\VV(p_1,\ldots,p_N)$ is a function of $N-1$
variables $q_j$, $q_j:=p_j-p_N$, $j=1,\dots,N-1$,
$q_j\in X$, $X=\{p-q:\,p,q\in S\}$,
$\VV(p_1,\ldots,p_N)=\Tilde\VV(q_1,\ldots,q_{N-1})$,
and is analytic in $q_j$. Hence, if $\det \VV(p_1,\ldots,p_N)=0$ for any choice of $p_j\in S$,
then, by Lemma~\ref{lem2}, $\det \Tilde\VV(q_1,\ldots,q_{N-1})$
vanishes on an open set and, due to the analyticity,
$\det \VV(p_1,\ldots,p_N)=0$ for any  $N$-tuple
$(p_1,\ldots,p_N)$ in $(\RR^{2})^N$.

On the other hand, by (c) we have $\VV_{jj}=\hat\nu(0)<0$
independently on $(p_1,\ldots,p_N)$ for every $j$.
Using the determinant expansion
and the condition (d) we obtain $\det\VV(p'_1,\ldots,p'_{N-1},p_N)\to \hat\nu(0) \det\VV(p'_1,\ldots,p'_N)\ne0$
as $p_N= r(\cos\alpha,\sin\alpha)$ and $r\to\pm\infty$.
This contradiction concludes the proof. 
\end{proof}

The item (d) in Lemma \ref{zero} is, of course, the finest one among the assumptions,
and there is an extended literature discussing the decay rates of the Fourier transforms
of measures, see e.~g.~\cite{brand,litt,Mar}.
For our purposes it is sufficient to consider the following two cases.

\begin{lem}\label{c1}
Let $\Gamma$ be a compact $C^2$ curve with non-vanishing curvature, and $\nu=\delta_\Gamma$,
then the assuption (d) in Lemma \ref{zero} is satisfied.
\end{lem}

\begin{proof}
Let $[0,l]\ni t\mapsto \big(x(t), y(t)\big)$ be the natural parametrization of $\Gamma$.
As the curvature does not vanish, one has 
\begin{equation}
        \label{eq-curv}
\big|x'(t)y''(t)-x''(t)y'(t)\big|\ge c>0
\quad \text{for all }t.
\end{equation}
Consider the expression
\[
\Hat\nu(p)=\dfrac{1}{2\pi}\int_0^l \exp \Big[ -i r \big( \cos\alpha x(t)+\sin\alpha y(t)
\big)\Big] dt, \quad p=r(\cos\alpha,\sin\alpha).
\]
By \eqref{eq-curv}, there exists $\delta>0$ for which one can divide $\Gamma$
into several smooth pieces such that one of the following
conditions is satisifed on each of them:
\begin{itemize}
\item $\big|\cos\alpha x'(t)+\cos\alpha y'(t)\big|\ge\delta$,
\item $\big|\cos\alpha x''(t)+\cos\alpha y''(t)\big|\ge\delta$.
\end{itemize}
Hence, by the van der Corput lemma on oscillatory integrals,
\cite[Section VIII.1, Proposition 2]{stein}, one has $\Hat\nu(p)\to 0$ for $r\to\infty$.
\end{proof}

\begin{lem}\label{c2}
Let $\Gamma$ be a line segment, and $\nu=\delta_\Gamma$,
then the assumption (d) in Lemma \ref{zero} is satisfied.
\end{lem}

\begin{proof}
Let $[0,1]\ni t\mapsto (x_0+at, y_0+bt)$ be a parametrization of $\Gamma$.
Writing
\begin{multline*}
\Hat\nu(p)=\dfrac{\sqrt{a^2+b^2}}{2\pi}
\int_0^1 \exp \Big[ -i r \big( \cos\alpha (x_0+at)+\cos\alpha 
(y_0+bt) \big)\Big]\,dt\\
=
\dfrac{\sqrt{a^2+b^2}\,\exp \big( -i r ( x_0 \cos\alpha +y_0 \sin\alpha ) \big)}{2\pi}
\int_0^1 \exp \Big[ -i r t\big( a \cos\alpha +b \sin\alpha \big)\Big]\,dt,\\
\quad p=r(\cos\alpha,\sin\alpha)
\end{multline*}
we immediately see that $\Hat\nu(r\cos\alpha,r\sin\alpha)\to 0$
as $r\to\pm\infty$ for any $\alpha$ with $a \cos\alpha +b \sin\alpha\ne 0$
by the Riemann-Lebesgue theorem.
\end{proof}

Now we can use Lemmas \ref{zero}, \ref{c1}, and \ref{c2} to prove a general result
on singular interactions supported by curves.

\begin{thm}\label{th2}
Assume that $S$ contains a $C^1$ arc which is not a line segment.
Let $\Gamma$ be a $C^2$ curve, $m=\delta_\Gamma$, and the restriction of $h$ to $\Gamma$
be a negative continuous function, then the operator $H$
has infinitely many eigenvalues below the essential spectrum.
\end{thm}

\begin{proof}
There exists a part of $\Gamma$, $\Gamma'$, with the following properties:
\begin{itemize}
\item $\Gamma'$ is a compact $C^2$ curve,
\item $\Gamma'$ either has  non-vanishing curvature or is a line segment,
\item $h|_{\Gamma'}\le - c$, $c>0$.
\end{itemize}
Represent $\nu=\nu_1+\nu_2$, $\nu_1=-c\delta_{\Gamma'}$,
$\nu_2:=\nu-\nu_1$. By construction one has $\nu_2\le 0$.
By Lemma \ref{zero}, for any $N$ there exist $p_1,\dots,p_N\in S$
such that the matrix $\big(\Hat\nu_1(p_j-p_k)\big)$ is negative definite.
As $\nu_2$ is non-positive, the matrix $\big(\Hat\nu_2(p_j-p_k)\big)$
is at least negative semi-definite by the Bochner theorem,
and $\big(\Hat\nu(p_j-p_k)\big)$ is hence negative definite.
Thus, $H$ has infinitely many eigenavlues below the essential spectrum by Theorem \ref{th1}.
\end{proof}

The assumption (a)  of Lemma \ref{zero} obviously holds for the Rashba and Dresselhaus Hamiltonians.
Hence, Theorem \ref{th2} guarantees the existence of infinitely many eigenvalues
below the continuous spectrum under negative perturbations supported by smooth curves.

If the set  $S$ for the unperturbed Hamiltonian $H_0$
is very ``bad'' and does not contain any smooth arc, then
an analogue of Lemma~\ref{lem2} becomes a difficult problem from the general topology.
Nevertheless, the assumptions of Lemma~\ref{lem2} are naturally satisfied
in reasonable examples including the above Rashba and Dresselhaus Hamiltonians.

We note in conclusion that the method proposed for estimating the number of eigenvalues
is quite universal but very rough; it does not take into account e.g. the Kramers degeneracy
(all eigenvalues of perturbed Rashba Hamiltonians must be at least twice degenerate).

\subsection*{Acknowledgments}
The work was supported  by the Deutsche Forschungsgemeinschaft.
We thank Johannes Brasche, Grigori Raikov, and Grigori Rozenblum
for useful advices and several bibliographical hints.

\end{document}